\documentclass{article}

\usepackage{color}

\usepackage{amsmath,amssymb,amsfonts,amsthm}

\usepackage{graphics} 
\usepackage{graphicx}
\usepackage{dsfont}

\usepackage{psfrag}

\newtheorem{teo}{Theorem}

\newtheorem{proposition}[teo]{Proposition}
\newtheorem{corollary}[teo]{Corollary}
\newtheorem{lemma}[teo]{Lemma}
\usepackage{colortbl}
    \definecolor{gray1}{rgb}{0,0,.8}
    \definecolor{gray2}{rgb}{0,1,0}
    \definecolor{gray3}{rgb}{0.8,0,.2}

\renewcommand{\P}{\mathds{P}}

\newcommand{\argmin}[1]{\underset{#1}{\mathrm{argmin\,}}}

\renewcommand{\natural}{\mathbb{N}}

\newcommand{\integernonnegative}{\mathbb{Z}_{\ge 0}}

\newcommand{\R}{\mathbb{R}} 
\newcommand{\N}{\mathbb{N}}  

\newcommand{\G}{\mathcal{G}} 
\newcommand{\V}{\mathcal{V}} 
\newcommand{\E}{\mathcal{E}} 

\newcommand{\eps}{\varepsilon}
\newcommand{\card}[1]{|#1|}  	

\newcommand{\Exp}{\mathds{E}} 

\newcommand{\1}{\mathbf{1}} 

\title{Almost sure convergence of a randomized algorithm for relative localization in sensor networks}%
\author{Chiara~Ravazzi\thanks{Chiara Ravazzi is with the Department of Electronics and Telecommunications (DET),
        Politecnico di Torino, Italy. E-mail: {chiara.ravazzi@polito.it}} \and Paolo~Frasca\thanks{Paolo Frasca is with the Department of Mathematical Sciences (DISMA), Politecnico di Torino, Italy. E-mail: {paolo.frasca@polito.it}} 
\and Roberto~Tempo\thanks{Roberto~Tempo is with CNR-IEIIT, Politecnico di Torino, Italy. E-mail: roberto.tempo@polito.it}
\and Hideaki~Ishii\thanks{Hideaki~Ishii is with the Department of Computational Intelligence and Systems Science,        Tokyo Institute of Technology, Japan. E-mail:        {ishii@dis.titech.ac.jp}} 
}

\begin{document}

\maketitle

\begin{abstract}
This paper regards the relative localization problem in sensor networks. 
We study a randomized algorithm, which is based on input-driven consensus dynamics and involves pairwise ``gossip'' communications and updates.
Due to the randomness of the updates, the state of this algorithm ergodically oscillates around a limit value. Exploiting the ergodicity of the dynamics, we show that the time-average of the state almost surely converges to the least-squares solution of the localization problem. Remarkably, the computation of the time-average does not require the sensors to share any common clock. Hence, the proposed algorithm is fully distributed and asynchronous.
\end{abstract}


\section{Introduction}

We consider the problem of relative localization for sensor networks that can be described as follows. We assume to have a group of agents representing the nodes of a graph, and a vector, indexed over the agents and unknown to them. The agents are allowed to take noisy measurements of the differences between their vector entries and those of their neighbors in the graph. The estimation problem consists in reconstructing the original vector, up to an additive constant. While an optimal solution can be easily found by a centralized least-squares approach, we are interested in finding effective \textit{distributed} solutions. More precisely, a solution is said to be distributed if it requires each node to use information which is available only at the node itself or from its immediate neighbors.
Following this approach, we have recently proposed~\cite{CR-PF-HI-RT:13a} a randomized ``gossip'' algorithm for distributed relative localization.
This algorithm, which is inspired by a gradient descent approach, involves the activation of a randomly chosen pair of neighboring nodes at each time step.
In our previous work, we have already given a convergence result for the algorithm: the mean-square error between the time-average of the states and the optimal solution asymptotically goes to zero.

In this paper, we study the algorithm using tools from ergodic theory and we obtain a related convergence result: the time-average of the states converges almost surely to the optimal solution. Significantly, our definition of time-average does not require the agents to be aware of any global clock or of any global variable which counts the number of interactions occurring on the network.

\subsection{State of the art}
The abstract problem of relative localization is deemed to have important applications in sensor and robotic networks~\cite{PB-JPH:07}. The proposed applications cover 
synchronization of uncertain clocks~\cite{AG-PRK:06a},  
as well as spacial localization~\cite{MD-JL-DR-ST:04} in mobile robotic networks when no absolute position information is available.
In the control literature, the problem has been popularized by~\cite{PB-JPH:07,PB-JPH:09,PB-JPH:08}. Distributed algorithms involving synchronized updates by the nodes have been proposed in the last few years~\cite{PB-JPH:07,AG-PRK:06a}. The latter paper uses a gradient-descent technique to solve the problem: this approach has been later followed in~\cite{SB-SDF-LS-DV:10,WSR-PF-FF:12}, as well as in our previous work~\cite{CR-PF-HI-RT:13a}.
Recently, papers have started to consider randomized asynchronous and randomized algorithms to solve the localization problem~\cite{RC-LS:12,RC-EDE-SZ:11,NMF-AZ:12}, see~\cite{RT-GC-FD:12} for details about randomized algorithms.

\subsection{Contribution}
In this work, we prove an ergodic theorem for the algorithm in~\cite[Eq.~(10)]{CR-PF-HI-RT:13a} and we show that a suitable time-averaging operation removes the persistent random oscillations which affect the ``raw'' estimates obtained through gossip communication and updates. The resulting time-averaged state converges {\em almost surely} to the optimal solution of the relative localization problem.

The ergodicity analysis of seemingly non-convergent dynamics started to attract attention in the theory of multi-agent systems and distributed control only recently.  
For example, we remark that the algorithm under study presents strong similarities with randomized algorithms for the PageRank computation, which have been recently proposed in several papers~\cite{HI-RT:10,HI-RT-EWB:12a,HI-RT-EWB:12b}. Indeed, randomized PageRank dynamics have been shown to converge modulo time-averaging: available results cover both convergence in mean square~\cite{HI-RT:10} and almost sure~\cite{WZ-HC-HF:13}. We stress, however, that our proof of almost sure convergence is based on completely different tools from~\cite{WZ-HC-HF:13}, which uses techniques from stochastic approximation.  
Related applications of ergodic theory can be found in the context of social networks~\cite{GC-FF:10,DA-GC-FF-AO:11}, where the authors show the ergodicity of specific opinion dynamics, which extend the well-known consensus dynamics to incorporate external influences and heterogeneous (stubborn) agents. Our analysis owes much to these techniques. 

Furthermore, we stress that ergodicity is a key property in enabling our framework to easily accommodate time-averages. Hence, we study an algorithm which is fully asynchronous and distributed. On the contrary, prior work on randomized algorithms assumed time-averages to be computed using a global iteration counter, see~\cite{HI-RT:10,WZ-HC-HF:13}.

\subsection{Organization of the paper}
We formally present the problem of relative localization in Section~\ref{sect:localization-problem}. Then, in Section~\ref{sect:algo-main-results} we define the randomized algorithm and state the main convergence results. These results are proved by the analysis presented in Section~\ref{sect:analysis}. We conclude the paper with some remarks on future research in the final section. 

\section{The problem of relative localization}\label{sect:localization-problem}
We consider a set of nodes $\mathcal{V}$ of cardinality $N$, endowed with an unknown scalar quantity $\bar x_v$ for $v\in\mathcal{V}$.
The relative localization problem
 consists, for each node $u\in \mathcal{V}$, in estimating the scalar value $\bar x_u$, based on noisy measurements of differences $\bar x_u-\bar x_v$ with certain neighbors $v$. 
An oriented
graph $\mathcal{G}=(\mathcal{V},\mathcal{E})$ is used to represent the
available measurements. 
The orientation of the pairs is conventionally assumed to be such that $(v,u)\in\mathcal{E}$ only if $u<v$.
We let $A\in\{0,\pm1\}^{\mathcal{E}\times\mathcal{V}}$ be the incidence matrix of the graph $\mathcal{G}$, which is defined as 
$$
A_{ew}=\begin{cases}
+1&\text{if } e= (v,w)\\
-1&\text{if } e=(w,v)\\
0&\text{otherwise}\\
\end{cases}
$$ for every $e\in \mathcal{E}$. 
We let $b\in \mathbb{R}^{\mathcal{E}} $ be the vector collecting the measurements
$$
b=A \bar x+\eta,
$$
where $\eta\in\mathbb{R}^{\mathcal{E}}$ is random noise with $\Exp[\eta]=0$ and $\Exp[\eta\eta^\top ]=\sigma^2I$, and $I$ is the identity matrix.
We define the set of the optimal estimates in a least squares sense as
\begin{equation}\label{LS}
X=\argmin{z\in\R^N}\|Az-b\|_2^2
\end{equation}
 where $\|\cdot\|_2$ is the Euclidean norm.
The set $X$ is described in the following well-known lemma \cite{PB-JPH:07}.
\begin{lemma}[Centralized solution]\label{lemma:centralized-LS} 
Let the graph $\mathcal{G}$ be connected and let $L:=A^\top A$ denote the Laplacian of the graph.
The following facts hold:
\begin{enumerate}
\item $x \in X$ if and only if $A^\top Ax =A^\top b$;
\item there exists a unique $x^{\star}\in X$ such that $||x^{\star}||_2=\min_{z\in X}||z||_2$;
\item $x^{\star}=L^{\dag} A^\top b$,
where $L^{\dag} $ denotes the Moore-Penrose pseudo-inverse of the Laplacian 
$L$.
\end{enumerate}
\end{lemma}
Note that $x^\star$ is the minimum-norm element of the affine space of solutions of~\eqref{LS}. Indeed, $A x^\star=A(x^\star+c\1)$ for all scalar $c$, where $\1$ denotes the vector of ones.

In this work, we are interested in designing randomized asynchronous algorithms to solve this problem: in view of the result above, we shall assume from now on that $\G$ is {\em connected}.

\section{Algorithm description and main results}\label{sect:algo-main-results}
This section is devoted to describe the algorithm which was proposed in~\cite[Section~V]{CR-PF-HI-RT:13a} and to state our main results.

The algorithm involves for each node $v\in \V$ a triple of states $(x_v,\kappa_v,\tilde x_v)$, depending on a discrete time index $k\in \integernonnegative$: these three variables play the following roles: $x_u(k)$ is the ``raw'' estimate of $\bar x_v$ obtained by $v$ at time $k$ through communications with its neighbors, $\kappa_v(k)$ counts the number of updates performed by $v$ up to time $k$, and $\tilde x_v(k)$ is the ``smoothed'' estimate obtained through time-averaging.
The algorithm is defined by choosing a scalar parameter $\gamma\in(0,1)$ and a sequence of random variables ${\theta(k)}_{k\in \integernonnegative}$ taking values in $\E$. 
At each time $k$, provided that $\theta(k)=(u,v)$, the state updates are performed according to the following rules:
the estimates evolve as 
\begin{subequations} \label{dyn2}
\begin{equation} \label{dyn2a}
\begin{split}
x_u(k+1)&=(1-\gamma)x_u(k)+ \gamma x_v(k)+ \gamma b_{(u,v)}\\
x_v(k+1)&=(1-\gamma)x_v(k)+ \gamma x_u(k)-\gamma b_{(u,v)}\\
x_w(k+1)&=x_w(k)\qquad \text{if }w\notin\{u,v\};\\
\end{split}
\end{equation}
the local times as
\begin{equation}
\begin{split}
\kappa_u(k+1)&=\kappa_u(k)+1\\
\kappa_v(k+1)&=\kappa_v(k)+1\\
\kappa_w(k+1)&=\kappa_w(k)\qquad \text{if }w\notin\{u,v\};\\
\end{split}
\end{equation}
and the time-averages as
\begin{equation}
\begin{split}\label{dyn2b}
\widetilde x_u(k+1)&=\frac1{\kappa_u(k+1)}\big( \kappa_u(k) \widetilde x_u(k) + x_u(k+1)\big)\\
\widetilde x_v(k+1)&=\frac1{\kappa_v(k+1)}\big( \kappa_v(k) \widetilde x_v(k) + x_v(k+1)\big)\\
\widetilde x_w(k+1)&=\widetilde x_w(k)\qquad \text{if }w\notin\{u,v\}.
\end{split}
\end{equation}
\label{eq:gossip-algo-ki}
\end{subequations}
We assume the sequence $\{\theta(k)\}_{k\in \integernonnegative}$ to be i.i.d., and its probability distribution to be uniform, {\it i.e.},
\begin{equation}\label{eq:uniform}
\P[\theta(k)=(u,v)]=\frac{1}{|\mathcal{E}|}, \qquad\forall k\in\integernonnegative,
\end{equation}
where $\card{\E}$ denotes the cardinality of $\E$. Note that this choice is made without loss of generality and convenience: the same approach may accommodate other distributions, as required by the application.

It should be noted that the time index $k$ in fact counts the number of updates which have occurred in the network, whereas for each $u\in \mathcal{V}$ the variable $\kappa_u(k)$ is the number of updates involving $u$ up to the current time. Hence, $\kappa_u$ is a local variable which is inherently known to agent $u$, even in case the common clock $k$ is unavailable.
This algorithm is totally asynchronous and fully distributed, in the sense that the time averaging process does not need the nodes to be aware of a common clock. This point provides a major improvement with respect to our previous work~\cite[Theorem~3]{CR-PF-HI-RT:13a}.

The dynamics in \eqref{dyn2a} oscillates persistently and fails to converge in a deterministic sense, as shown in Figure~\ref{fig:simul-complete}. 
\begin{figure}[h]
\begin{center}
\includegraphics[width=.8\columnwidth]{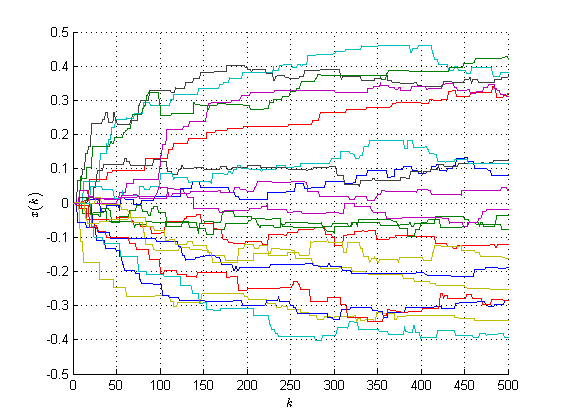}
\includegraphics[width=.8\columnwidth]{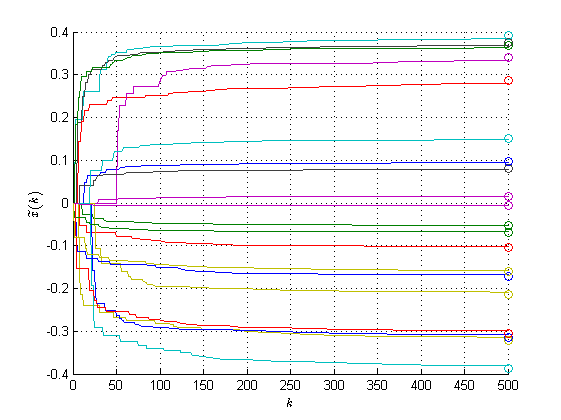}
\caption{Complete graph, $N=20$, $\gamma=0.1$, zero initial conditions.}
\label{fig:simul-complete}
\end{center}
\end{figure}
However, the oscillations concentrate around the solution of the least squares problem, as it is formally stated in the following two results. 
The first result regards the behavior of the average dynamics.
\begin{proposition}[Convergence in expectation]\label{prop:1st-order} Consider the dynamics~\eqref{dyn2a} with uniform edge selection~\eqref{eq:uniform}. 
Then,
\begin{align*}
\lim_{k\rightarrow+\infty}\Exp[x (k)]=x^{\star}.
\end{align*}
\end{proposition}
The second result, instead, states that the sample dynamics is well-represented by the average one, {\it i.e.}, the process is ergodic. 
\begin{teo}[Almost sure convergence of ergodic means]\label{as_convergence} The dynamics in \eqref{dyn2a}, with uniform edge selection~\eqref{eq:uniform} and $x(0)=0$, is {\em ergodic}. If $\{n_\ell\}_{\ell\in \mathbb{N}}$ is a sequence of nonnegative integers, then with probability one
$$
\lim_{k\rightarrow\infty}\frac{1}{k}\sum_{\ell=1}^{k} x(n_\ell){=}x^{\star}.
$$
\end{teo}

From these two facts, which are proved in the next section, we immediately  deduce the following statement, which motivates the definition of the algorithm  previously introduced indicates that $\widetilde x_u(k)$ is the right variable'' to approximate the optimal estimate $x^{\star}_u.$
\begin{corollary}
The dynamics in~\eqref{dyn2b} is such that $$\lim_{k\to+\infty} \widetilde{x}(k)=x^{\star}$$
with probability one. \end{corollary}

\section{Ergodicity analysis}\label{sect:analysis}
In this section we study the convergence properties of the vector $ x(k)$ for the update model described in \eqref{dyn2a}. As shown in Fig. \ref{fig:simul-complete}, the estimate of each agent oscillates persistently; on the other hand the time averages approach the optimal solution $x^{\star}$ in Lemma \ref{lemma:centralized-LS}. Although the process $x(k)$ almost surely fails to converge, we prove that it converges in distribution to a random variable $x_{\infty}$ and is ergodic. 

To begin, we rewrite the dynamics of~\eqref{dyn2a} as
\begin{equation}\label{eq:dyn3-old}
x(k+1)=Q(k)x(k)+y(k)
\end{equation}
where $$Q(k)=I-\gamma(e_u-e_v)(e_u-e_v)^\top, $$ vector $e_u$ denotes the vector of the canonical basis corresponding to $u$, and $$y(k)=b_{\theta(k)}(e_u-e_v)$$ provided $\theta(k)=(u,v)$ with $k\in\integernonnegative$.
Consequently, we prove the following result either by direct computation or by using Lemma~5 in~\cite{CR-PF-HI-RT:13a}.
\begin{lemma}\label{lemma: expected_value} For the distribution~\eqref{eq:uniform},  it holds
\begin{gather*}
\Exp[Q(k)]=I-\gamma\frac{L}{|\mathcal{E}|},
\qquad
\Exp\left[ y(k)\right] = \gamma\frac{A^\top b}{|\mathcal{E}|}.
\end{gather*}
\end{lemma}

We note that  for all $k$ the matrix $Q(k)$ is doubly stochastic and the sum of the elements in $y(k)$ is zero, that is,
\begin{gather}
\begin{split}\label{d_stoch}
\1^\top Q(k)=\1,\quad Q(k)\1=\1,\quad\1^\top y(k)=0.
\end{split}
\end{gather}
In particular, if the vector $x(0)$ is initialized to zero, then $\1^\top x(k)=0$ for each $k\in\integernonnegative$.

These observations imply that the dynamics of $x(k)$ is equivalently described by the iterate
\begin{equation}\label{eq:dyn3}
x(k+1)=P(k)x(k)+y(k),
\end{equation}
where $P(k)=Q(k)(I-\frac1n\1\1^\top)=(I-\frac1n\1\1^\top)Q(k)$ is a projection of $Q(k)$ outside the ``consensus sub-space'' spanned by $\1$. This rewriting is instrumental to study the convergence behavior of the process $\{x(k)\}_{k\in\integernonnegative}$, by asymptotic techniques of iterated random functions, which we recall from~\cite{PD-DF:99}. These techniques require, in order to study the random process~\eqref{eq:dyn3}, to consider the associated {\em backward} process $\overleftarrow{x}(k)$, which we define below.

For any time instant $k$, consider the random matrices $P{(k)}$ and ${y}{(k)}$ and define the matrix product
\begin{equation}\label{P}
\overrightarrow{P}(\ell,m):=P{(m)}P{(m-1)}\cdots P{(\ell+1)}P{(\ell)}
\end{equation}
with $\ell\in\{0,\ldots,m\}.$
Then, the iterated affine system in~\eqref{eq:dyn3} can be rewritten as
\begin{equation}
 x(k+1)=\overrightarrow{P}(0,k)  x(0)+\sum_{0\leq \ell\leq k}\overrightarrow{P}(\ell+1,k)y(\ell).
\end{equation}
The corresponding \emph{backward process} is defined by
\begin{equation}
\overleftarrow{x}(k+1)=\overleftarrow{P}(0,k)  x(0)+\sum_{0\leq \ell\leq k}\overleftarrow{P}(0,\ell-1)y(\ell),
\end{equation}
where
\begin{equation}\label{Preverse}
\overleftarrow{P}(\ell,m):=P{(\ell)}P{(\ell+1)} \cdots P{(m-1)} P{(m)} 
\end{equation}
with $\ell\in\{0,\ldots,m\}.$
Crucially, the backward process $\overleftarrow{x}(k)$ has at every time $k\in\integernonnegative$ the same probability distribution of $x(k)$. 
The main tool to study the backward process is the following well-known result. 
\begin{lemma}[Theorem~2.1 in~\cite{PD-DF:99}]\label{Diaco}
Let us consider the Markov chain $\{X(k)\}_{k\in\N}$ defined by
$$X({k+1})={A(k)}X(k) + B(k)\quad k\in \integernonnegative$$ 
where $A(k)\in\R^{\V\times\V}$ and $B(k)\in\R^{\V}$ are i.i.d. random variables.
Let us assume that
\begin{equation}
\Exp[\log\|A(k)\|]<\infty\qquad
\Exp[\log\|B(k)\|]<\infty.
\end{equation}
The corresponding backward random process $\overleftarrow{X}(k)$ converges a.s. to a finite limit $X_{\infty}$ if and only if
\begin{equation}\label{eq:ip_Diaco}
\inf_{k>0}\frac{1}{k}\Exp\left[\log\|A(1)\ldots A(k)\|\right]<0.
\end{equation}
If \eqref{eq:ip_Diaco} holds, the distribution of $X_{\infty}$ is the unique invariant distribution for the Markov chain $X(k)$.
\end{lemma}

This result provides conditions for the backward process to converge to a limit. Although the forward process has a different behavior compared to the backward process, the forward and backward processes have the same distribution. This fact allows us to determine, by studying the backward process $\overleftarrow{x}(k)$, whether the sequence of random variables $\{x(k)\}_{k\in\integernonnegative}$ converges in distribution to the invariant distribution of the Markov chain in~\eqref{eq:dyn3}.

%
%
%
%
This analysis is done in the following result.
\begin{lemma}\label{lem:back_convergence}
Consider the random process $x(k)$ defined in~\eqref{eq:dyn3}, where $P(k)$ and $y(k)$ are i.i.d.. Then 
$\overleftarrow{x}(k)$ converges a.s. to a finite limit $x_{\infty}$, and the distribution of $x_{\infty}$ is the unique invariant distribution for $x(k)$.
\end{lemma}
\begin{proof}
In order to apply Lemma~\ref{Diaco}, let us compute
\begin{align*}
&\inf_{k\in\N}\frac{1}{k}\Exp\left[\log\|\overleftarrow{P}(0,k-1)\|_{1 }\right]\\
&\qquad\leq \inf_{k\in\N}\frac{1}{k}\log\Exp\left[\|\overleftarrow{P}(0,k-1)\|_{1}\right]\\
& \qquad = \inf_{k\in\N}\frac{1}{k}\log\Exp\left[\max_{w\in \V}\sum_{i\in \V}(\overleftarrow{P}(0,k-1))_{v w}\right]\\
& \qquad\leq\inf_{k\in\N}\frac{1}{k}\log\Exp\left[\sum_{w\in \V}\sum_{v\in \V}(\overleftarrow{P}(0,k-1))_{v w}\right]\\
& \qquad\leq\inf_{k\in\N}\frac{1}{k}\log\sum_{w\in \V}\sum_{v\in \V}\Exp\left[\overleftarrow{P}(0,k-1)_{v w}\right]\\
& \qquad\leq\inf_{k\in\N}\frac{1}{k}\log\left(n\left\|\Exp\left[\overleftarrow{P}(0,k-1)\right]\right\|_{\infty}\right)\\
& \qquad =\inf_{k\in\N}\frac{1}{k}\log\left(n\left\|\prod_{\ell=0}^{k-1}\Exp\left[{P}(\ell)\right]\right\|_{\infty}\right).
\end{align*}
Let $q$ be the number of distinct eigenvalues of $\Exp[P(k)]$, denoted as 
$\{\lambda_{\ell}\}_{\ell=1}^{q}$, and consider the Jordan canonical decomposition $
\Exp\left[P(k)\right]=UJU^{-1}
$. Then $\left\|\prod_{s=0}^{k-1}\Exp\left[{P}(s)\right]\right\|_{\infty}\leq \|U\|_{\infty}\|J^k\|_{\infty}\|U^{-1}\|_{\infty}$. 
Since the $k$-th power of the Jordan block of size $s$ is 
\begin{align*}&\begin{bmatrix}\lambda & 1 & 0 & \cdots & 0\\
0 & \lambda & 1 & \cdots & 0\\
\vdots & & \ddots & & \vdots\\
0 & \cdots & 0 & \lambda & 1\\
0 & 0 & \cdots  & & \lambda
 \end{bmatrix}^k\\
&\qquad=
\begin{bmatrix}\lambda^k & {k \choose 1} \lambda^{k-1} & {k \choose 2} \lambda^{k-2} & \cdots  & {k \choose s-1} \lambda^{k-s+1}\\ 
0 & \lambda^k & {k \choose 1} \lambda^{k-1}& \cdots  & {k \choose s-2} \lambda^{k-s+2}\\
\vdots & & \ddots & & \vdots\\
0 & \cdots & 0 & \lambda^k & {k \choose 1} \lambda^{k-1}\\
0 & 0 & \cdots  & & \lambda^k
\end{bmatrix},\end{align*}
we deduce that 
$$\|J^k\|_{\infty }=\max_{v\in \V}\sum_{w\in \V}(J^k)_{v w}= \max_{\ell=1, \dots, q}\sum_{m=0}^{s_\ell-1}\lambda_\ell^{k-m}{k\choose m},$$
where $s_\ell$ is the size of the largest  Jordan block corresponding to $\lambda_\ell.$
Then
\begin{align*}\|J^k\|_{\infty }&\leq \max_{\ell=1, \dots, q}|\lambda_\ell|^k\sum_{m=0}^{s_\ell-1}|\lambda_\ell|^{-m}{k\choose m}\\
&\leq\max_{\ell=1, \dots, q}|\lambda_\ell|^kk^{n}\sum_{m=0}^{s_\ell-1}|\lambda_\ell|^{-m}\\
&\leq \chi\rho^kk^{n},\end{align*}
where $\chi$ is a constant independent of $k$ and 
$$
\rho= \left\| \Exp[P(k)]\right\|_{2}.
$$
From Lemma~\ref{lemma: expected_value}, it follows that $\rho<1$.
We conclude that there exists a constant $C=\|U\|_\infty \|U^{-1}\|_\infty \chi$, independent of $k$, such that
$$\Exp\left[\log\|\overleftarrow{P}(0,k-1)\|_{1}\right]\leq\log\left(nC\rho^kk^{n}\right).
$$
and, consequently,
\begin{align}
\nonumber
&\inf_{k\in\N}\frac{1}{k}\Exp\left[\log\|\overleftarrow{P}(0,k-1)\|_{1}\right]\\
\label{eq:bound-for-Diaco}
&\qquad\leq\lim_{k\rightarrow\infty}\frac{\log(C n k^{n}\rho^k)}{k}\\
\nonumber&\qquad=\log\rho<0.
\end{align}
The claim then follows from Lemma~\ref{Diaco}.
\end{proof}

As a consequence, we deduce that also the (forward) random process $x(k)$, defined in~\eqref{eq:dyn3}, converges in distribution to a limit $x_{\infty}$, and the distribution of $x_{\infty}$ is the unique invariant distribution for $x(k)$.
%
The oscillations of~\eqref{dyn2a} are {\em ergodic}, as stated in the following result. 
\begin{lemma}\label{lemma:ergodic}
The random process $x(k)$ defined in~\eqref{eq:dyn3} is {\em ergodic}: if $\{n_\ell\}_{\ell\in \natural}$ is a sequence of nonnegative integers, then with probability one
$$
\lim_{k\rightarrow\infty}\frac{1}{k}\sum_{\ell=1}^{k} x(n_\ell){=}\Exp\left[x_{\infty}\right].
$$
Furthermore, $\Exp[x_\infty]=L^{\dag}A^{\top}b$.
\end{lemma}

\begin{proof} 
We begin by showing the ergodicity. Let $z(0)$ be a random vector independent from $x(0)$ with the same distribution as $x_{\infty}$. Let $\{z(k)\}_{k\in \integernonnegative}$ be the sequence such that
$$
z(k)=\overrightarrow{P}(0,k-1)  z(0)+\sum_{0\leq \ell\leq k-1}\overrightarrow{P}(\ell+1,k-1)y(\ell)
$$
where $\overrightarrow{P}(\ell+1,k-1)$ is defined as in \eqref{P}.
It should be noted that, since the process $z (k)$ is stationary, we can apply the law of large numbers and immediately conclude that, with probability one,
$$\lim_{k\rightarrow\infty}\frac{1}{k}\sum_{s=0}^{k-1}z(s)=\Exp[x_\infty].$$
On the other hand, we compute 
\begin{align*}
&\mathbb{P}\left(\|x(k)-z(k)\|_1\geq \eps^k\right)\\
&\qquad\leq\frac{\Exp\left[\|\overrightarrow{P}(0,k-1)  (z(0)-x(0))\|_1\right]}{\eps^k}\\
&\qquad\leq \frac{\Exp\left[\|\overrightarrow{P}(0,k-1)\|_1\|z(0)-x(0))\|_1\right]}{\eps^k}\\
& \qquad\leq \frac{\Exp\left[\|\overrightarrow{P}(0,k-1)\|_1\right]\Exp\left[\|z(0)-x(0)\|_1\right]}{\eps^k}\\
& \qquad\leq \frac{Cnk^{n}\rho^{k}}{\eps^k} \Exp\left[\|z(0)-x(0)\|_1\right],
\end{align*}
where we have used~\eqref{eq:bound-for-Diaco}.
If we choose $\eps\in(\rho,1)$, then the Borel-Cantelli Lemma \cite[Theorem~1.4.2]{VB:95} implies that with probability one $\|x(k)-z(k)\|_1<\eps^k$ for all but finitely many values
of $k\geq0$. Therefore, almost surely $\frac{1}{k}\sum_{s=0}^{k-1}\|x(s)-z(s)\|_1$ converges
 to zero as $k$ goes to infinity, 
$$\lim_{k\rightarrow\infty}\frac{1}{k}\sum_{s=0}^{k-1}x(s)=\Exp[x_\infty].$$
The statement follows when we observe that the same arguments hold if we replace all summations over the nonnegative integers with summations over a subsequence of nonnegative integers.

To complete the proof, we only need to compute the expectation of $x_\infty$. Using the independence among $P(k)$s and $y(k)$s, we obtain
\begin{align*}
&\Exp[x(k)]=
\Exp[{P}] ^{k}x(0)+\sum_{0\leq \ell\leq k-1} \Exp[{P(k)}] ^{k-\ell-1}\Exp[y(k)]
\end{align*}
and we conclude from Proposition~\ref{prop:1st-order} (see also~\cite[Proposition~6]{CR-PF-HI-RT:13a}) that $$\Exp[x_\infty]=\lim_{k\to+\infty}\Exp[x(k)]=L^{\dag}A^\top b.$$
\end{proof}

\section{Concluding remarks}
In this paper, we have studied a randomized gossip algorithm for the relative localization problem, which is distributed and asynchronous. Using tools from ergodic theory, we have shown the almost sure convergence of the algorithm to the optimal solution. Future research will seek broader applications of these techniques in multi-agents systems.



\begin{thebibliography}{10}
\providecommand{\url}[1]{#1}
\csname url@samestyle\endcsname
\providecommand{\newblock}{\relax}
\providecommand{\bibinfo}[2]{#2}
\providecommand{\BIBentrySTDinterwordspacing}{\spaceskip=0pt\relax}
\providecommand{\BIBentryALTinterwordstretchfactor}{4}
\providecommand{\BIBentryALTinterwordspacing}{\spaceskip=\fontdimen2\font plus
\BIBentryALTinterwordstretchfactor\fontdimen3\font minus
  \fontdimen4\font\relax}
\providecommand{\BIBforeignlanguage}[2]{{%
\expandafter\ifx\csname l@#1\endcsname\relax
\typeout{** WARNING: IEEEtran.bst: No hyphenation pattern has been}%
\typeout{** loaded for the language `#1'. Using the pattern for}%
\typeout{** the default language instead.}%
\else
\language=\csname l@#1\endcsname
\fi
#2}}
\providecommand{\BIBdecl}{\relax}
\BIBdecl

\bibitem{CR-PF-HI-RT:13a}
C.~Ravazzi, P.~Frasca, H.~Ishii, and R.~Tempo, ``A distributed randomized
  algorithm for relative localization in sensor networks,'' in \emph{{E}uropean
  {C}ontrol {C}onference}, 2013, to appear.

\bibitem{PB-JPH:07}
P.~Barooah and J.~P. Hespanha, ``Estimation from relative measurements:
  {A}lgorithms and scaling laws,'' \emph{{IEEE} Control Systems Magazine},
  vol.~27, no.~4, pp. 57--74, 2007.

\bibitem{AG-PRK:06a}
A.~Giridhar and P.~R. Kumar, ``Distributed clock synchronization over wireless
  networks: Algorithms and analysis,'' in \emph{{IEEE} Conference on Decision
  and Control}, San Diego, CA, USA, Dec. 2006, pp. 4915--4920.

\bibitem{MD-JL-DR-ST:04}
D.~Moore, J.~Leonard, D.~Rus, and S.~Teller, ``Robust distributed network
  localization with noisy range measurements,'' in \emph{International
  {C}onference on {E}mbedded {N}etworked {S}ensor {S}ystems}, ser. SenSys '04, New York,
  NY, USA, 2004, pp. 50--61.

\bibitem{PB-JPH:09}
P.~Barooah and J.~P. Hespanha, ``Error scaling laws for linear optimal
  estimation from relative measurements,'' \emph{IEEE Transactions on
  Information Theory}, vol.~55, no.~12, pp. 5661--5673, 2009.

\bibitem{PB-JPH:08}
------, ``Estimation from relative measurements: {E}lectrical analogy and large
  graphs,'' \emph{IEEE Transactions on Signal Processing}, vol.~56, no.~6, pp.
  2181--2193, 2008.

\bibitem{SB-SDF-LS-DV:10}
S.~Bolognani, S.~D. Favero, L.~Schenato, and D.~Varagnolo, ``Consensus-based
  distributed sensor calibration and least-square parameter identification in
  {WSN}s,'' \emph{International Journal of Robust and Nonlinear Control},
  vol.~20, no.~2, pp. 176--193, 2010.

\bibitem{WSR-PF-FF:12}
W.~S. Rossi, P.~Frasca, and F.~Fagnani, ``Transient and limit performance of
  distributed relative localization,'' in \emph{{IEEE} Conference on Decision
  and Control}, Maui, HI, USA, Dec. 2012, pp. 2744--2748.

\bibitem{RC-LS:12}
\BIBentryALTinterwordspacing
R.~Carli and L.~Schenato, ``Exponential-rate consensus-based algorithms for
  estimation from relative measurements,'' University of Padova, Tech. Rep.,
  Sep. 2012. [Online]. Available:
  \url{http://automatica.dei.unipd.it/people/schenato/publications.html}
\BIBentrySTDinterwordspacing

\bibitem{RC-EDE-SZ:11}
R.~Carli, E.~{D'Elia}, and S.~Zampieri, ``A {PI} controller based on asymmetric
  gossip communications for clocks synchronization in wireless sensors
  networks,'' in \emph{{IEEE} Conf. on Decision and Control and European
  Control Conference}, Orlando, FL, USA, 2011.

\bibitem{NMF-AZ:12}
N.~M. Freris and A.~Zouzias, ``Fast distributed smoothing of relative
  measurements,'' in \emph{{IEEE} Conference on Decision and Control}, Maui,
  HI, USA, Dec. 2012, pp. 1411--1416.

\bibitem{RT-GC-FD:12}
R.~Tempo, G.~Calafiore, and F.~Dabbene, \emph{Randomized Algorithms for
  Analysis and Control of Uncertain Systems, with Applications}.\hskip 1em plus
  0.5em minus 0.4em\relax Springer, 2005.

\bibitem{HI-RT:10}
H.~Ishii and R.~Tempo, ``Distributed randomized algorithms for the {PageRank}
  computation,'' \emph{IEEE Transactions on Automatic Control}, vol.~55, no.~9,
  pp. 1987--2002, 2010.

\bibitem{HI-RT-EWB:12a}
H.~Ishii, R.~Tempo, and E.~W. Bai, ``A web aggregation approach for distributed
  randomized {PageRank} algorithms,'' \emph{IEEE Transactions on Automatic
  Control}, vol.~57, no.~11, pp. 2703--2717, 2012.

\bibitem{HI-RT-EWB:12b}
------, ``Pagerank computation via a distributed randomized approach with lossy
  communication,'' \emph{Systems \& Control Letters}, vol.~61, no.~12, pp.
  1221--1228, 2012.

\bibitem{WZ-HC-HF:13}
W.-X. Zhao, H.-F. Chen, and H.-T. Fang, ``Almost sure convergence of
  distributed randomized {PageRank} computation,'' \emph{IEEE Transactions on
  Automatic Control}, 2013, submitted.

\bibitem{GC-FF:10}
G.~Como and F.~Fagnani, ``Scaling limits for continuous opinion dynamics
  systems,'' \emph{The Annals of Applied Probability}, vol.~21, no.~4, pp.
  1537--1567, 2011.

\bibitem{DA-GC-FF-AO:11}
D.~Acemoglu, G.~Como, F.~Fagnani, and A.~Ozdaglar, ``Opinion fluctuations and
  disagreement in social networks,'' \emph{Mathematics of Operations Research},
  2012, in press.

\bibitem{PD-DF:99}
P.~Diaconis and D.~Freedman, ``Iterated random functions,'' \emph{SIAM Review},
  vol.~41, pp. 45--76, 1999.

\bibitem{VB:95}
V.~Borkar, \emph{Probability Theory: An Advanced Course}.\hskip 1em plus 0.5em
  minus 0.4em\relax Springer, 1995.

\end{thebibliography}

\end{document}